%% file: main.tex
\documentclass[11pt]{article}
\usepackage{graphicx} 
\usepackage{amsmath,amssymb,amsthm,amsfonts}
\usepackage{bm,fullpage,hyperref}
\usepackage{color}
\usepackage{booktabs,multirow}
\usepackage{comment}
\usepackage{apxproof}

\usepackage{algpseudocode}
\usepackage[linesnumbered,ruled,vlined]{algorithm2e}
\SetKwProg{Procedure}{Procedure}{}{}

\allowdisplaybreaks

\newcommand{\E}{\mathop{\mathbf{E}}}

\newcommand{\floor}[1]{\left \lfloor{#1}\right \rfloor}

\newcommand{\OPT}{\mathrm{OPT}}

\newcommand{\CONN}{\mathrm{conn}}
\newcommand{\LCA}{\mathrm{lca}}

\newcommand{\SHAPMST}{\mathrm{Shap}_\mathrm{MST}}
\newcommand{\SHAPMATCH}{\mathrm{Shap}_\mathrm{Match}}

\newtheorem{theorem}{Theorem}[section]
\newtheorem{lemma}[theorem]{Lemma}

\theoremstyle{definition}
\newtheorem{definition}[theorem]{Definition}
\newtheorem{example}[theorem]{Example}

\title{Lipschitz Continuous Allocations for Optimization Games}
  \author{Soh Kumabe\footnote{Supported by JST, PRESTO Grant Number JPMJPR192B and JSPS KAKENHI Grant Number JP21J20547. Most of this work was done while the author was a student at The University of Tokyo.}\\
  CyberAgent\\
  \texttt{kumabe\_soh@cyberagent.co.jp}
  \and
  Yuichi Yoshida\footnote{Supported by JST, PRESTO Grant Number JPMJPR192B and JSPS KAKENHI Grant Number JP20H05965.}\\
  National Institute of Informatics\\
  \texttt{yyoshida@nii.ac.jp}}
\date{May 2024}

\begin{document}

\maketitle

\begin{abstract}
In cooperative game theory, the primary focus is the equitable allocation of payoffs or costs among agents.
However, in the practical applications of cooperative games, accurately representing games is challenging.
In such cases, using an allocation method sensitive to small perturbations in the game can lead to various problems, including dissatisfaction among agents and the potential for manipulation by agents seeking to maximize their own benefits. Therefore, the allocation method must be robust against game perturbations.

In this study, we explore optimization games, in which the value of the characteristic function is provided as the optimal value of an optimization problem. To assess the robustness of the allocation methods, we use the Lipschitz constant, which quantifies the extent of change in the allocation vector in response to a unit perturbation in the weight vector of the underlying problem. Thereafter, we provide an algorithm for the matching game that returns an allocation belonging to the $\left(\frac{1}{2}-\epsilon\right)$-approximate core with Lipschitz constant $O(\epsilon^{-1})$. Additionally, we provide an algorithm for a minimum spanning tree game that returns an allocation belonging to the $4$-approximate core with a constant Lipschitz constant.

The Shapley value is a popular allocation that satisfies several desirable properties. Therefore, we investigate the robustness of the Shapley value. We demonstrate that the Lipschitz constant of the Shapley value for the minimum spanning tree is constant, whereas that for the matching game is $\Omega(\log n)$, where $n$ denotes the number of vertices.

\end{abstract}




\input{intro}

\input{prelim}

\input{matching_new}

\input{mst}

\input{shapley}

\bibliographystyle{abbrv}
\bibliography{main}

\appendix

\input{omitted_proofs}

\end{document}

%% file: intro.tex
\section{Introduction}

\subsection{Background and Motivation}

Cooperative games model decision-making scenarios in which multiple agents can achieve greater benefits through cooperation. A primary concern in cooperative game theory is the allocation of payoffs or costs provided by the grand coalition in an acceptable manner to each agent. Among the cooperative games, those defined by optimization problems corresponding to the set of agents involved are known as \emph{optimization games}~\cite{deng1999algorithmic}.

Consider the following well-known examples formulated by the matching game (MG)~\cite{eriksson2001stable,benedek2023complexity} and the minimum spanning tree game (MSTG)~\cite{claus1973cost}:
\begin{description}
    \item[MG:] Members of a tennis club form pairs for a doubles tournament. Each pair of players has a predicted value for the prize money they would win if they teamed up. How should the total prize money won by all pairs be distributed among members?
    \item[MSTG:] Multiple facilities cooperate to construct a power grid to receive electricity from a power plant. Each potential power line has a predetermined installation cost. When constructing a power grid that ensures electricity distribution to all facilities, what cost should each facility bear?
\end{description}
In these examples, the characteristic function of the game often contains errors, or can be manipulated through deliberate misreporting. For instance,
\begin{description}
\item[MG:] It is challenging to accurately predict the compatibility of pairs who have never teamed up before. Moreover, pairs known to work well may hide this fact.
\item[MSTG:] Costs for power line installation might be misestimated owing to unforeseen terrain or geological conditions caused by natural disasters or inadequate surveys, land rights, or landscape regulations. Facilities may conceal these issues.
\end{description}
In such uncertain situations, allocations that drastically change with slight perturbations in the game can lead to problems, such as
\begin{description}
\item[MG:] If minor estimation errors in predicted prize money significantly change the benefits for each player, players might not be satisfied with the allocation. In addition, when someone falsely reports a substantial increase in their gain, it becomes difficult to prove the intent of the manipulation if the degree of falsehood is small.
\item[MSTG:] Facilities might not accept the allocation if trivial issues in power line construction significantly increase their cost burden. Moreover, if such issues are used to significantly reduce their own costs or increase those of competing facilities, several minor construction issues may be concealed, leading to risk management problems.
\end{description}
To avoid these issues, it is desirable to use allocations that are robust against perturbations in real-world cooperative games.

Before proceeding, we define some terms in cooperative game theory. The \emph{cooperative game} $(V,\nu)$ is defined as a pair consisting of a set of \emph{agents} $V$ and a \emph{characteristic function} $\nu\colon 2^V \to \mathbb{R}_{\geq 0}$ representing the payoff or cost obtained when subsets of agents cooperate.
An \emph{optimization game} is defined by the optimization problem $\mathcal{P}$ in a discrete structure. The \emph{$\mathcal{P}$-game} $(G,w)$ is defined from a pair of structures $G$ consisting of \emph{agent set} $V$, \emph{edge set} $E$, and a weight vector $w\in \mathbb{R}_{\geq 0}^{E}$. In several games, $G$ is a graph $(V,E)$.
For each subset $S \subseteq V$, the characteristic function value $\nu(S)$ is defined as the optimum value of $\mathcal{P}$ on the substructure corresponding to $S$ (e.g., the subgraph induced by $S$) with respect to weight vector $w$.
An optimization game is a \emph{welfare allocation game} (resp., \emph{cost allocation game}) if each agent intends to maximize (resp., minimize) the value allocated to it. 


Now, we formally introduce matching games and minimum spanning tree games, as discussed in previous examples.
\begin{definition}[Matching game~\cite{benedek2023complexity,deng1999algorithmic}]
Let $G=(V,E)$ be an undirected graph. For a vertex set $S\subseteq V$ and an edge weight vector $w \in \mathbb{R}_{\geq 0}^{E}$, let $\OPT(S,w)$ denote the maximum matching weight in $G[S]$ with respect to weight $w$, where $G[S]$ represents the subgraph of $G$ induced by $S$. The \emph{matching game} of $G$ with respect to weight $w$ is defined as $(V,\nu)$, where $\nu(S)=\OPT(S,w)$.
The matching game is a welfare-allocation game.
\end{definition}

\begin{definition}[Minimum spanning tree game~\cite{bird1976cost,claus1973cost}]
Let $G=(V\cup \{r\},E)$ be an undirected graph. For a vertex set $S\subseteq V$ and edge weight vector $w\in \mathbb{R}_{\geq 0}^{E}$, let $\OPT(S,w)$ denote the minimum weight of a spanning tree in $G[S\cup \{r\}]$ with respect to weight $w$.
Notably, vertex $r$ does not correspond to the agent.
To ensure that the characteristic function has finite values, we assume that edge $(r,v)$ exists for all $v \in V$.
The \emph{minimum spanning tree game} of $G$ with respect to weight $w$ is defined as $(V,\nu)$, where $\nu(S)=\OPT(S,w)$.
The minimum spanning tree game is a cost-allocation game.
\end{definition}

Let us return to the discussion of the robustness of the allocation.
To measure the robustness of the allocation, we introduce the concept of \emph{Lipschitz continuity} in the allocation of these games, analogous to the introduction of Kumabe and Yoshida~\cite{kumabe2023lipschitz} for discrete optimization problems under the same name. 
Algorithm $\mathcal{A}$ that takes a weight vector $w\in \mathbb{R}^{E}$ and returns an allocation $x\in \mathbb{R}_{\geq 0}^{V}$ is \emph{$L$-Lipschitz} or has \emph{Lipschitz constant} $L$ if: 
\begin{align}
    \sup_{\substack{w,w'\in \mathbb{R}_{\geq 0}^{E}\\w\neq w'}}\frac{\|\mathcal{A}(w)-\mathcal{A}(w')\|_1}{\|w-w'\|_1}\leq L.\label{eq:lip_game}
\end{align}
If the Lipschitz constant of the allocation method is small, the change in the allocation in response to a unit change in the weight vector is guaranteed to be small. Employing such an allocation method can resolve these problems as follows:
\begin{description}
    \item[MG:] Minor mistakes in estimating the compatibility of pairs will not significantly affect the overall distribution of prize money, making it easier for agents to accept the allocation. Additionally, substantial misreporting is necessary to increase benefits significantly through declaration adjustments, making schemes more likely to be exposed.
    \item[MSTG:] Minor issues related to the installation of power lines will not cause significant fluctuations in the cost burden for each facility, making it easier for them to accept their share of the costs. Additionally, because the loss incurred by reporting such issues is small, the likelihood of these issues being concealed is reduced.
\end{description}


\subsubsection{The core.}
The \emph{core} is the most fundamental solution concept in cooperative game theory. An allocation vector $x\in \mathbb{R}^{V}$ is in the core of (welfare allocation) cooperative game $(V,\nu)$ if:
\begin{align}
    \sum_{v\in S}x_v \geq \nu(S)\quad (S\subsetneq V),\quad \sum_{v\in V}x_v = \nu(V).\label{eq:welfare_approx}
\end{align}
Similarly, an allocation vector $x\in \mathbb{R}^{V}$ is in the core of (cost allocation) cooperative game $(V,\nu)$ if:
\begin{align}
    \sum_{v\in S}x_v \leq \nu(S)\quad (S\subsetneq V),\quad \sum_{v\in V}x_v = \nu(V).\label{eq:cost_approx}
\end{align}

Because the core is one of the most fundamental solution concepts, it is natural to desire a Lipschitz continuous algorithm $\mathcal{A}$ that takes a weight vector as the input and returns an allocation belonging to the core. However, this can only be achieved in a limited number of situations for two reasons. First, a non-empty core does not necessarily exist in all games. Second, even for games in which the core exists in all instances, there is no guarantee that a vector from the core can be selected such that the Lipschitz constant of $\mathcal{A}$ remains small.

To examine the second reason, we consider the \emph{assignment game} introduced by Shapley and Shubik~\cite{shapley1971assignment}, which is a special case of the matching game in which the underlying graph is bipartite and known to have a non-empty core~\cite{deng1999algorithmic,shapley1971assignment}.

\begin{example}\label{ex:pathmatch}
Let $n$ be an odd number greater than or equal to $5$. Let us consider a path $G=(V,E)$ consisting of $n$ vertices labeled sequentially as $v_1, \dots, v_n$. The weight vectors $w, w' \in \mathbb{R}_{\geq 0}^{E}$ are defined as follows:
\begin{align*}
    w_{(v_i,v_{i+1})}=1 \quad (i=1,\dots, n-1),
    \qquad
    w'_{(v_i,v_{i+1})}=\begin{cases}
        1 & \text{if }2\leq i\leq n-2\\
        0 & \text{otherwise}
    \end{cases}
\end{align*}
In this case, the allocations $x, x'$ belonging to the core of the assignment game for $w$ and $w'$ are both unique and obtained as follows:
\begin{align*}
    x_i=\begin{cases}
        1 & \text{if $i$ is even}\\
        0 & \text{otherwise}             
    \end{cases}
    \qquad
    x'_i=\begin{cases}
        1 & \text{if $i$ is odd and $i\not \in \{1,n\}$}\\
        0 & \text{otherwise}
    \end{cases}
\end{align*}
If $\mathcal{A}$ is an algorithm that takes a weight vector and returns a vector in the core of the game, then it must satisfy:
\begin{align*}
    \frac{\|\mathcal{A}(w)-\mathcal{A}(w')\|_1}{\|w-w'\|_1}=\frac{n-2}{2}=\Omega(n)
\end{align*}
Thus, algorithm $\mathcal{A}$ must have a large Lipschitz constant, $\Omega(n)$.
\end{example}

Therefore, we consider providing a Lipschitz continuous algorithm that outputs allocations that satisfy some looser solution concepts. The least core~\cite{shapley1966quasi} is one such solution concept that can be defined even for games in which a core does not exist. However, this does not resolve the problem of a large Lipschitz constant (for instance, in the game in Example~\ref{ex:pathmatch}, the core and the least core coincide, resulting in a Lipschitz constant of $\Omega(n)$). Instead, we consider the approximate core, which is a solution concept that multiplicatively relaxes the constraints of partial coalitions.

\subsubsection{Approximate Core.}

The approximate core was introduced by Faigle and Kern~\cite{faigle1993some} as a useful solution concept for games where the core is empty. Intuitively, the approximate core represents the core when the constraints for partial coalitions are relaxed by a factor $\alpha$.
For $\alpha\leq 1$, the allocation vector $x\in \mathbb{R}^{V}$ is in the \emph{$\alpha$-(approximate) core} of the (welfare allocation) cooperative game $(V,\nu)$ if:
\begin{align*}
    \sum_{v\in S}x_v \geq \alpha \nu(S)\quad (S\subsetneq V),\quad \sum_{v\in V}x_v = \nu(V).
\end{align*}
Similarly, for $\alpha\geq 1$, the allocation vector $x\in \mathbb{R}^{V}$ is in the \emph{$\alpha$-(approximate) core} of the (cost allocation) cooperative game $(V,\nu)$ if:
\begin{align*}
    \sum_{v\in S}x_v \leq \alpha \nu(S)\quad (S\subsetneq V),\quad \sum_{v\in V}x_v = \nu(V).
\end{align*}

In this study, we provide algorithms with small Lipschitz constants that return an $\alpha$-approximate core for some constant $\alpha$ in several optimization games. 
For the matching game, we obtain the following:
\begin{theorem}\label{thm:matching}
Let $\epsilon\in \left(0,\frac{1}{2}\right]$. For the matching game, there is a polynomial-time algorithm with a Lipschitz constant $O(\epsilon^{-1})$ that returns $\left(\frac{1}{2}-\epsilon\right)$-approximate core allocation.
\end{theorem}
Note that Faigle and Kern~\cite{faigle1993some} and Vazirani~\cite{vazirani2022general} showed that the $\frac{2}{3}$-approximate core of the matching game is non-empty, and that this is also a tight bound. However, their allocation, which is constructed using the optimal solution of the matching LP, is not Lipschitz continuous. Our allocation compromises the core approximability by $\frac{1}{6}+\epsilon$ to ensure Lipschitz continuity.

Our allocation was inspired by the Lipschitz continuous algorithm for the maximum weight-matching problem proposed by Kumabe and Yoshida~\cite{kumabe2023lipschitz}. Similarly, the proposed algorithm is based on the greedy method. We emphasize our results because their algorithm reduces the approximation ratio to $\frac{1}{8}-\epsilon$ to obtain a Lipschitz constant of $O(\epsilon^{-1})$. Although the definitions of Lipschitz continuity and the approximation ratio for discrete algorithms differ from those for allocations, making a simple comparison infeasible, our analysis is simpler and offers better approximation guarantees.


For the minimum spanning tree game, we have the following.
\begin{theorem}\label{thm:mst}
For the minimum spanning tree game, there is a polynomial-time algorithm with Lipschitz constant $O(1)$ that returns $4$-approximate core allocation.
\end{theorem}
The core of the minimum spanning tree game is non-empty; an allocation by Bird~\cite{bird1976cost} is known to belong to the core and can be computed in polynomial time. However, it is not Lipschitz continuous. As in the matching game, our allocation compromises the core approximability to ensure Lipschitz continuity.

The original motivation for the approximate core was to provide useful solution concepts for games with non-empty cores. Therefore, no studies have been conducted on the approximate core of the minimum spanning tree game. In this regard, our setting demonstrates the usefulness of considering an approximate core, even for games with a non-empty core.

\subsubsection{Shapley Value.}
Let $\mathfrak{S}_V$ be the set of all permutations over $V$.
For $\sigma\in \mathfrak{S}_V$ and $v\in V$, let $x_{\sigma,v}=\nu(\{\sigma(1),\dots, \sigma(k)\})-\nu(\{\sigma(1),\dots, \sigma(k-1)\})$, where $k$ is the integer with $\sigma(k)=v$.
The \emph{Shapley value}~\cite{shapley1965balanced} of game $(V,\nu)$ is the vector $s$ defined by:
\begin{align*}
    s_v = \frac{1}{|V|!}\sum_{\sigma\in \mathfrak{S}_V}x_{\sigma,v}.
\end{align*}

The Shapley value does not necessarily belong to the core, even if it exists, and the computation is \#$\mathsf{P}$-hard in most optimization games. However, they exhibit various desirable properties and have a wide range of applications~\cite{hart1989shapley}. Therefore, investigating the Lipschitz continuity of Shapley values is a natural task. 

Considering its various properties, it is natural to expect the Shapley value to always have a small Lipschitz constant for general optimization games. Conversely, given the difficulty in computing it, it is natural to anticipate that it may not have a bounded Lipschitz constant. However, neither was accurate. Specifically, we demonstrate that whether the Shapley value in optimization games has a small Lipschitz constant depends on the game.
\begin{theorem}\label{thm:shap_matching}
There is a graph $G$ such that the Shapley value of the matching game on $G$ has a Lipschitz constant $\Omega(\log n)$.
\end{theorem}

\begin{theorem}\label{thm:shap_mst}
The Shapley value of the minimum spanning tree game has a Lipschitz constant of $2$.
\end{theorem}
Notably, the computation of the Shapley value for the minimum spanning tree game is \#$\mathsf{P}$-hard~\cite{ando2012computation}. The result of Theorem~\ref{thm:shap_mst} is particularly interesting because it shows a value that is computationally difficult to calculate may still have a small Lipschitz constant.

\subsection{Related Work}

\subsubsection{Optimization Games.}
The history of optimization games begins with the \emph{assignment game} proposed by Shapley and Shubik~\cite{shapley1971assignment}. They showed that the core of the assignment game is represented by the optimal solution of dual linear programming and is always non-empty.
Deng, Ibaraki, and Nagamochi~\cite{deng1999algorithmic} defined a class of games in which the characteristic function is represented by integer programming, and discussed the core structures of such games. In particular, for the matching game, they proved that the core is not always non-empty, but that the core non-emptiness problem, core membership problem, and that of generating a vector in the core if it is non-empty are solvable in polynomial time.
Aziz and de Keijzer~\cite{aziz2013shapley} proved that computing the Shapley value of the matching game is \#$\mathsf{P}$-hard.
Additionally, studies on the generalization of matching games such as the \emph{hypergraph matching game}~\cite{conitzer2006complexity,deng1999algorithmic,kumabe2020hyper} and \emph{$b$-matching game}~\cite{biro2018stable,kumabe2020b,sotomayor1992multiple,xiao2021approximate}, are being conducted.
For a more detailed survey of the algorithmic aspects of matching games, see~\cite{benedek2023complexity}.

The minimum spanning tree game was proposed by Claus and Kleitman~\cite{claus1973cost}. Bird~\cite{bird1976cost} later proposed an allocation defined as follows: Regarding the minimum spanning tree of a given graph as a rooted tree rooted at $r$, each agent corresponding to a vertex $v$ is allocated a cost equal to the weight of the edge from $v$ to its parent. Granot and Huberman~\cite{granot1981minimum} proved this allocation is in the core. Unfortunately, it is not (Lipschitz) continuous.
Faigle et al.~\cite{faigle1997complexity} proved that the core membership problem for the minimum spanning tree game is $\mathsf{coNP}$-hard.
Ando~\cite{ando2012computation} proved that the Shapley value of the minimum spanning tree game is \#$\mathsf{P}$-hard.

The concept of an approximate core was introduced by Faigle and Kern~\cite{faigle1993some} as a useful solution for games where the core is empty. Their work discussed the approximate core allocations for several optimization games, including the matching game. In particular, for the matching game, they constructed a $\frac{2}{3}$-core allocation based on LP relaxation.
Subsequently, extensive studies have been conducted to determine the best $\alpha$ for which the $\alpha$-core is always non-empty in various optimization games in which the core can be empty, such as the traveling salesman game~\cite{faigle1998approximately,faigle1993some} and bin-packing games ~\cite{faigle1993some,faigle1998approximate,kern2012integrality,kuipers1998bin,qiu2016approximate,woeginger1995rate}.
Recently, following the rediscovery of Faigle and Kern's results by Vazirani~\cite{vazirani2022general}, approximate cores have been derived for several optimization games such as the $b$-matching game~\cite{xiao2021approximate} and the edge cover game~\cite{lu2023approximate}.


Notably, in the definition by Faigle and Kern~\cite{faigle1993some} and in subsequent papers, the right-hand side of Equations~\eqref{eq:welfare_approx} or~\eqref{eq:cost_approx} being $(1 \pm \epsilon)\nu(S)$ is referred to as the $\epsilon$-core. However, following the conventions in the field of discrete optimization and the study by Vazirani~\cite{vazirani2022general} and subsequent studies ~\cite{lu2023approximate,xiao2021approximate}, we adopt the definitions of Equations~\eqref{eq:welfare_approx} and~\eqref{eq:cost_approx}. An $\epsilon$-core in our definition corresponds to a $|1-\epsilon|$-core in Faigle and Kern's definition.

\subsubsection{Lipschitz Continuity of Discrete Algorithms.}

Inspired by recent studies on the average sensitivity~\cite{hara2023average,kumabe22average,peng2020average,varma2023average,yoshida2022average} for \emph{unweighted} discrete optimization problems, Kumabe and Yoshida~\cite{kumabe2023lipschitz} introduced the Lipschitz continuity of a randomized algorithm $\mathcal{A}$ for \emph{weighted} discrete optimization problems as follows:
\begin{align}
    \sup_{\substack{w,w' \in \mathbb{R}_{\geq 0}^V,\\w\neq w'}}\frac{\min_{\mathcal{D}\in \Pi(\mathcal{A}(G,w),\mathcal{A}(G,w'))}\E_{(S,S')\sim \mathcal{D}}\left[d_{\mathrm{w}}\left((S,w),(S',w')\right)\right]}{\|w-w'\|_1},\label{eq:lip_discrete}
\end{align}
where $\mathcal{A}(G,w)$ represents the output distribution of algorithm $\mathcal{A}$ for input $G$ with weight vector $w\in \mathbb{R}_{\geq 0}^{V}$, $\Pi(X,X')$ denotes the set of all joint distributions for random variables $X$ and $X'$, and
\begin{align*}
    d_{\mathrm{w}}((S,w),(S',w')) = \left\|\sum_{v \in S}w_v \bm{1}_v - \sum_{v \in S'}w'_v \bm{1}_v\right\|_1 
    = \sum_{v \in S \cap S'}|w_v - w'_v| + \sum_{v \in S \setminus S'}w_v + \sum_{v \in S' \setminus S}w'_v.
\end{align*}
They also proposed algorithms with small Lipschitz constants for the minimum spanning tree, shortest path, and maximum weight matching problems. In a subsequent study~\cite{kumabe2023cover}, the authors obtained algorithms with small Lipschitz constants for the minimum weight vertex cover, minimum weight set cover, and feedback vertex set problems.

In discrete optimization, the outputs of the algorithms are discrete sets. Thus, deterministic algorithms cannot be Lipschitz continuous; they consider randomized algorithms and adopt the earth mover's distance in the numerator of~\eqref{eq:lip_discrete}. In contrast, the outputs of our algorithms are allocations of continuous values, allowing deterministic algorithms to be Lipschitz continuous and randomized algorithms to be derandomized by taking the expectation. Therefore, in our setting, we can use a simple definition of the deterministic algorithms expressed in Equation~\eqref{eq:lip_game}.



\subsection{Organization}

The remainder of this paper is organized as follows: In Section~\ref{sec:basic}, we provide several useful lemmas to analyze Lipschitz continuity of allocations. In Section~\ref{sec:matching}, we prove Theorem~\ref{thm:matching} by providing a Lipschitz continuous polynomial-time algorithm that returns an approximate core allocation to the matching game. In Section~\ref{sec:mst}, we prove Theorem~\ref{thm:mst} by providing a Lipschitz continuous polynomial-time algorithm that returns an approximate core allocation for the minimum spanning tree game. 
Finally, in Section~\ref{sec:shapley}, we prove Theorems~\ref{thm:shap_matching} and~\ref{thm:shap_mst} to discuss the Lipschitz continuity of the Shapley value. 

%% file: prelim.tex
\section{Basic Facts}\label{sec:basic}

In this section, we provide useful lemmas to analyze the core approximability and Lipschitz continuity of our allocations.
The following lemma is useful for obtaining the Lipschitz constant: 
We omit this proof because it is similar to Lemma~1.7 of~\cite{kumabe2023lipschitz}.
\begin{lemma}\label{lem:seeoneelement}
Let $(G=(V,E),w)$ be an optimization game and $\mathcal{A}$ be an algorithm that takes a pair $(G,w)$ and outputs an allocation.
Suppose that there exist some $c>0$ and $L>0$ such that
\begin{align*}
    \left\|(\mathcal{A}(G,w), \mathcal{A}(G,w+\delta \bm{1}_e)\right\|_1\leq \delta L
\end{align*}
holds for any $e\in E$, $w\in \mathbb{R}_{\geq 0}^E$, and $\delta > 0$ with either $\delta\leq c\cdot w_e$ or $w_e=0$, where $\bm{1}_e$ represents the characteristic vector of $e$.
Then, $\mathcal{A}$ has a Lipschitz constant $L$.
\end{lemma}

In the analysis of the core approximability and Lipschitz continuity of allocations, we can simplify the discussion by ignoring the constraint $\sum_{v\in V}x_v = \nu(V)$ for a grand coalition. Therefore, we require a method to obtain an allocation with a bounded Lipschitz constant from vectors that satisfy only the constraints for partial coalitions with a bounded Lipschitz constant. 
To achieve this, we require a mild assumption in the game. An optimization game $(G=(V,E),w)$ is said to be \emph{reasonable} if the inequality $|\nu(V,w)-\nu(V,w')|\leq \|w-w'\|_1$ holds for all $w, w' \in \mathbb{R}{\geq 0}^{E}$. Notably, this is a fair assumption. For instance, games defined for optimization problems in the form of
\begin{align*}
    \text{$\max$ or $\min$} \sum_{e\in X}w_e, \text{ subject to } X\in \mathcal{F}
\end{align*}
for $\mathcal{F}\subseteq 2^E$ are all reasonable.
The next lemma applies to welfare allocation games, such as the matching game.
The proofs of the next two lemmas are given in Appendix~\ref{sec:omitted}.
\begin{lemma}\label{lem:benri_welfare}
Let $D\geq 1$. Let $\mathcal{A}$ be an algorithm that takes a weight vector $w\in \mathbb{R}_{\geq 0}^{E}$ and returns an allocation $\mathcal{A}(w)\in \mathbb{R}_{\geq 0}^{V}$ for a reasonable welfare allocation game. Assume $\mathcal{A}$ satisfies $\|\mathcal{A}(w)\|_1\leq D\nu(V,w)$, $\sum_{v\in S}\mathcal{A}(w)_v\geq \nu(S,w)$ for all weight vector $w$ and $S \subseteq V$, and $\|\mathcal{A}(w)-\mathcal{A}(w')\|_1\leq L\|w-w'\|_1$ for all two weight vectors $w$ and $w'$.
Then, there is an algorithm that returns $\frac{1}{D}$-approximate core allocation with Lipschitz constant $2L+1$.
\end{lemma}
The next lemma applies to cost allocation games, such as the minimum spanning tree game.
\begin{lemma}\label{lem:benri_cost}
Let $D\geq 1$. Let $\mathcal{A}$ be an algorithm that takes a weight vector $w\in \mathbb{R}_{\geq 0}^{E}$ and returns an allocation $\mathcal{A}(w)\in \mathbb{R}_{\geq 0}^{V}$ for a reasonable cost allocation game. Assume $\mathcal{A}$ satisfies $\|\mathcal{A}(w)\|_1\geq \nu(V,w)$, $\sum_{v\in S}\mathcal{A}(w)_v\leq D\nu(S,w)$ for all weight vector $w$, and $\|\mathcal{A}(w)-\mathcal{A}(w')\|_1\leq L\|w-w'\|_1$ for all two weight vectors $w$ and $w'$.
Then, there is an algorithm that returns $D$-approximate core allocation with Lipschitz constant $2L+1$.
\end{lemma}

%% file: matching_new.tex
\section{Matching Game}\label{sec:matching}
In this section, we prove Theorem~\ref{thm:matching} by giving a Lipschitz continuous algorithm that returns an approximate core allocation of the matching game.
We obtain the proof by constructing an algorithm that satisfies the assumptions in Lemma~\ref{lem:benri_welfare}. Specifically, we prove that an algorithm that returns a vector represented by
\begin{align}
\int_{0}^{1}\Call{MatchingGame}{G,w,b,\alpha}\mathrm{d}b\label{eq:match_out}
\end{align}
satisfies the assumptions of Lemma~\ref{lem:benri_welfare}, where the procedure \Call{MatchingGame}{} is provided in Algorithm~\ref{alg:matching}.
We give a deterministic algorithm to compute this integral in Section~\ref{sec:match_combine}.

\Call{MatchingGame}{$G,w,b,\alpha$} first rounds each edge weight $w_e$ to a value $\widehat{w}_e$ that is proportional to a power of $\alpha$, where the proportionality constant is determined by $b$. Thereafter, it sorts the edges in descending order of $\widehat{w}_e$ and greedily selects them to form maximal matching $M$. Finally, for each edge $e\in M$, the algorithm allocates $\widehat{w}_e$ to both endpoints of the edges in $M$.

\begin{algorithm}[t!]
\caption{Lipschitz continuous allocation for the matching game}\label{alg:matching}
\Procedure{\emph{\Call{MatchingGame}{$G,w,b,\alpha$}}}{
    \KwIn{A graph $G=(V,E)$, where the edge set is indexed with integers in $\{1,2,\ldots,|E|\}$, a weight vector $w \in \mathbb{R}_{\geq 0}^E$, $b\in [0,1]$, and $\alpha\in (1,2]$.}
    For each $e\in E$ with $w_e>0$, let $\widehat{w}_e\leftarrow \alpha^{i_e+1+b}$, where $i_e$ is the unique integer such that $\alpha^{i_e+b}\leq w_e < \alpha^{i_e+1+b}$\;
    $z\leftarrow 0^{V}$, $M\leftarrow \emptyset$\;
    \For{$e\in E$ in descending order of $w_e$, where ties are broken according to their indices}{\label{line:match_loop}
        \If{none of the two endpoints of $e$ are covered by $M$}{
            Add $e$ to $M$\;
            $z_v\leftarrow \widehat{w}_e$ for each endpoint $v$ of $e$\;
        }
    }
    \Return $z$\;
}
\end{algorithm}

\subsection{Core Approximability}

The proofs of the following two lemmas for the core approximability analysis of Algorithm~\ref{alg:matching} are relatively straightforward.
\begin{lemma}\label{lem:match_core1}
We have $\|z\|_1\leq 2\alpha\OPT(V,w)$.
\end{lemma}
\begin{proof}
Because $M$ is a matching of $G$, we have $\sum_{e\in M}w_e\leq \OPT(V,w)$. Because the modified weight $\widehat{w}_e$ of each edge $e$ in $M$ contributes twice to $\|z\|_1$, we obtain 
\[
    \|z\|_1 = 2\sum_{e\in M}\widehat{w}_e\leq 2\alpha \sum_{e\in M}w_e\leq 2\alpha\OPT(V,w).
    \qedhere
\]
\end{proof}

\begin{lemma}\label{lem:match_core2}
Let $S\subseteq V$. Then, we have $\sum_{v\in S}z_v\geq \OPT(S,w)$.
\end{lemma}
\begin{proof}
Let $e=(u,v)\in E$.
When edge $e$ begins to be examined in the loop starting from Line~\ref{line:match_loop}, if at least one of $u$ or $v$ (say, $u$) is already covered by $M$, then we have $\widehat{w}_e \leq z_u$. If neither $u$ nor $v$ are covered, then edge $e$ is added to $M$, resulting in $\widehat{w}_e = z_u = z_v$.  
Therefore, $\widehat{w}_e\leq \max\{z_u,z_v\}\leq z_u+z_v$.
Let $M'$ be the maximum matching of $G[S]$. Then, we have
\[
    \sum_{v\in S}z_v\geq \sum_{(u,v)\in M'}(z_u+z_v)\geq \sum_{e\in M'}\widehat{w}_e \geq \sum_{e\in M'}w_e  = \OPT(S,w).
    \qedhere
\]
\end{proof}

\subsection{Lipschitz Continuity}
Let $G=(V,E)$ be a graph, $f\in E$, $\delta>0$, $b\in [0,1]$, and $\alpha>1$. 
We will bound
\begin{align}
    \frac{1}{\delta}\int_{0}^{1}\left\|\Call{MatchingGame}{G,w,b,\alpha} - \Call{MatchingGame}{G,w+\delta \bm{1}_f,b,\alpha}\right\|_1 \mathrm{d}b,\label{eq:match_int}
\end{align}
which is an upper bound on
\begin{align*}
    \frac{1}{\delta}\left\|\int_{0}^{1}\Call{MatchingGame}{G,w,b,\alpha}\mathrm{d}b - \int_{0}^{1}\Call{MatchingGame}{G,w+\delta \bm{1}_f,b,\alpha}\mathrm{d}b\right\|_1.
\end{align*}
From Lemma~\ref{lem:seeoneelement}, bounding~\eqref{eq:match_int} for $\delta \leq w_f$ or $w_f=0$ is sufficient to prove Lipschitz continuity.
We denote the value of $\widehat{w}$, $M$, and $z$ in \Call{MatchingGame}{$G,w,b,\alpha$} (resp., \Call{MatchingGame}{$G,w+\delta \bm{1}_f,b,\alpha$}) as $\widehat{w}$, $M$, and $z$ (resp., $\widehat{w}'$, $M'$, and $z'$).

When $\widehat{w}_f=\widehat{w}'_f$, \Call{MatchingGame}{$G,w,b,\alpha$} and \Call{MatchingGame}{$G,w+\delta \bm{1}_f,b,\alpha$} output the same vector. Assume otherwise.
In \Call{MatchingGame}{$G,w,b$}, edge $e_1$ coming \emph{before} edge $e_2$ refers to $e_1$ being considered before $e_2$ in the loop starting from Line~\ref{line:match_loop}, and is denoted as $e_1\prec_{\widehat{w}} e_2$. 
In other words, either $\widehat{w}_{e_1} > \widehat{w}_{e_2}$ or $\widehat{w}_{e_1} = \widehat{w}_{e_2}$ and the index of $e_1$ comes earlier than that of $e_2$. For $e_1\neq f\neq e_2$, the relations $e_1\prec_{\widehat{w}} e_2$ and $e_1\prec_{\widehat{w}'} e_2$ are equivalent. Thus, we simply denote this as $e_1\prec e_2$.
The following lemma forms the core of our Lipschitzness analysis:
\begin{lemma}\label{lem:match_diff}
Assume $\widehat{w}_f\neq \widehat{w}'_f$. Then, we have $\|z-z'\|_1\leq 2\widehat{w}'_{f}$.
\end{lemma}
\begin{proof}
For each edge $e \neq f$ such that $e\in M'\setminus M$ (resp., $e\in M\setminus M'$), edge $g$ is a \emph{witness} of $e$ if it is adjacent to $e$ in $M$ (resp., $M'$) and $g\prec_{\widehat{w}} e$ (resp., $g\prec_{\widehat{w}'} e$). Intuitively, a witness of $e$ is the edge that directly causes $e$ to be excluded from $M$ or $M'$.

From this definition, the witness of $e\in M'\setminus M$ belongs to $M\setminus M'$ and vice versa. Because $\prec$ is an ordering on $E\setminus \{f\}$, by tracing the witnesses from any edge in $M\triangle M'$, we will consequently arrive at $f$. This implies that as long as $M\neq M'$, the edges in $M\triangle M'$ form a single path or cycle including $f$.
Moreover, in this case, we have $f\in M'\setminus M$, because if we can add $f$ to $M$ in \Call{MatchingGame}{$G,w,b,\alpha$}, we could also add $f$ to $M'$ in \Call{MatchingGame}{$G,w+\delta \bm{1}_f,b,\alpha$}.

Let us now complete the proof. When $M=M'$, we have $\|z-z'\|_1 \leq 2(\widehat{w}'_f-\widehat{w}_f) \leq 2\widehat{w}'_f$.
Otherwise, we let $f=(u_0,v_0)$. Then, there exists a unique maximal sequence of vertices $(u_0,\dots, u_k)$ such that $(u_i,u_{i+1})\in M\triangle M'$ for all $i$ and $f\prec_{\widehat{w}'} (u_0,u_1)\prec (u_1,u_2)\prec \cdots \prec (u_{k-1},u_k)$, and a unique maximal sequence of vertices $(v_0,\dots, v_l)$ such that $(v_i,v_{i+1})\in M\triangle M'$ for all $i$ and $f\prec_{\widehat{w}'} (v_0,v_1)\prec (v_1,v_2)\prec \cdots \prec (v_{l-1},v_l)$ (if $M\triangle M'$ forms a cycle, then $(u_{k-1},u_{k})=(v_{l},v_{l-1})$). Now, we have
\begin{align*}
    &\|z-z'\|_1 \leq \sum_{i=0}^{k}|z_{u_i}-z'_{u_i}| + \sum_{i=0}^{l}|z_{v_i}-z'_{v_i}|\\
    & = \left(|\widehat{w}'_f-\widehat{w}_{(u_0,u_1)}|+\sum_{i=1}^{k}|\widehat{w}_{(u_{i-1},u_i)}-\widehat{w}_{(u_{i},u_{i+1})}|\right) 
    + \left(|\widehat{w}'_f-\widehat{w}_{(v_0,v_1)}|+\sum_{i=1}^{l}|\widehat{w}_{(v_{i-1},v_i)}-\widehat{w}_{(v_{i},v_{i+1})}|\right)\\
    & \leq 2\widehat{w}'_f,
\end{align*}
where the last inequality is from the fact that the sequences defined by $\left(\widehat{w}'_f, \widehat{w}_{(u_0,u_1)},\dots, \widehat{w}_{(u_{k-1},u_{k})}\right)$ and $\left(\widehat{w}'_f, \widehat{w}_{(v_0,v_1)},\dots, \widehat{w}_{(v_{l-1},v_{l})}\right)$ are both decreasing.
\end{proof}

The following lemma analyzes the probability that $\widehat{w}_f\neq \widehat{w}'_f$ happens.
\begin{lemma}
If $b$ is sampled uniformly from $[0,1]$, $\widehat{w}_f\neq \widehat{w}'_f$ happens with a probability of at most $\frac{\delta}{w_f\ln \alpha}$.
\end{lemma}
\begin{proof}
$\widehat{w}_f\neq \widehat{w}'_f$ happens when there exists an integer $i$ with $w_f<\alpha^{i+b}\leq w_f+\delta$, indicating that $\floor{\log_\alpha w_f-b}\neq \floor{\log_\alpha w_f+\delta-b}$. This happens with probability 
\[
    \log_\alpha (w_f+\delta) - \log_\alpha w_f = \log_\alpha \left(1+\frac{\delta}{w_f}\right)\leq \frac{1}{\ln \alpha}\cdot \frac{\delta}{w_f}=\frac{\delta}{w_f \ln \alpha}.
    \qedhere
\]
\end{proof}
Now, we complete our Lipschitzness analysis.
\begin{lemma}\label{lem:match_lip}
We have
\begin{align*}
    \int_{0}^{1}\left\|\Call{MatchingGame}{G,w,b,\alpha} - \Call{MatchingGame}{G,w+\delta \bm{1}_f,b,\alpha}\right\|_1 \mathrm{d}b\leq \frac{12}{\alpha-1}\delta.
\end{align*}
\end{lemma}
\begin{proof}
If $w_f=0$, we have
\begin{align*}
    &\int_{0}^{1}\left\|\Call{MatchingGame}{G,w,b,\alpha} - \Call{MatchingGame}{G,w+\delta \bm{1}_f,b,\alpha}\right\|_1 \mathrm{d}b\\
    &\leq 2\widehat{w}'_f \leq 2\alpha (w_f+\delta) = 2\alpha \delta\leq \frac{12}{\alpha-1}\delta,
\end{align*}
where the last inequality is from $\alpha\leq 2$.
Otherwise, we have
\begin{align*}
    &\int_{0}^{1}\left\|\Call{MatchingGame}{G,w,b,\alpha} - \Call{MatchingGame}{G,w+\delta \bm{1}_f,b,\alpha}\right\|_1 \mathrm{d}b\\
    &\leq \frac{\delta}{w_f\ln \alpha}\cdot 2\widehat{w}'_f
    \leq \frac{\delta}{w_f\ln \alpha}\cdot 2\alpha(w_f+\delta) 
    \leq \frac{\delta}{w_f\ln \alpha}\cdot 4\alpha w_f = \frac{4\alpha\delta}{\ln \alpha}\leq \frac{12}{\alpha-1}\delta,    
\end{align*}
where the third inequality is from $\delta \leq w_f$ and the last inequality is from the fact that $\frac{\alpha}{\ln \alpha} \leq \frac{3}{\alpha-1}$ holds for $\alpha\in (1,2]$.
\end{proof}

\subsection{Proof of Theorem~\ref{thm:matching}}\label{sec:match_combine}
Combining Lemmas~\ref{lem:match_core1},~\ref{lem:match_core2}, and~\ref{lem:match_lip} and applying Lemma~\ref{lem:benri_welfare} yields the following:
\begin{lemma}\label{lem:match_final}
Let $\epsilon \in \left(0,\frac{1}{2}\right]$. For the matching game, an algorithm that returns $\left(\frac{1}{2}-\epsilon\right)$-approximate core allocation with Lipschitz constant $O(\epsilon^{-1})$ exists. 
\end{lemma}
\begin{proof}
Let $\alpha=1+2\epsilon$. 
By combining Lemmas~\ref{lem:match_core1},~\ref{lem:match_core2},~\ref{lem:match_lip}, and~\ref{lem:benri_welfare}, we obtain an algorithm that returns $\frac{1}{2\alpha}$-approximate core allocation for the matching game with Lipschitz constant $\frac{24}{\alpha-1}+1$.
As $\frac{1}{2(1+2\epsilon)}\geq \frac{1}{2}-\epsilon$ and $\frac{24}{2\epsilon}+1\leq O(\epsilon^{-1})$, this algorithm satisfies the claims of the lemma.
\end{proof}

\begin{proof}[Proof of Theorem~\ref{thm:matching}]
It is sufficient to prove that the allocation defined by Lemma~\ref{lem:match_final} can be computed in polynomial time.
For each edge $e \in E$, let $b_e = \log_\alpha w_e - \floor{\log_\alpha w_e}$, and sort the $b_e$ values in ascending order to obtain a sequence $t_1, \dots, t_{|E|}$. For convenience, we set $t_0 = 0$ and $t_{|E|+1} = 1$.
For each $i = 0, \dots, |E|$, the behavior of Algorithm~\ref{alg:matching} for any $b \in [t_i, t_{i+1})$ is identical, except for the constant multiplier on $\widehat{w}$. Therefore, by running Algorithm~\ref{alg:matching} for $b = t_i$ and appropriately scaling the result, and thereafter summing these results for each $i$, we can compute the integral in Equation~\eqref{eq:match_out} in polynomial time.
\end{proof}

%% file: mst.tex
\section{Minimum Spanning Tree Game}\label{sec:mst}
In this section, we prove Theorem~\ref{thm:mst} by giving a Lipschitz continuous algorithm that returns an approximate core allocation for the minimum spanning tree game.
The proof is obtained by constructing an algorithm that satisfies the assumption of Lemma~\ref{lem:benri_cost}. Specifically, we prove that an algorithm that returns a vector represented by
\begin{align*}
    \int_{0}^{1}\Call{MSTGame}{G,w,b}\mathrm{d}b
\end{align*}
satisfies the assumptions of Lemma~\ref{lem:benri_cost}, where the procedure \Call{MSTGame}{} is provided in Algorithm~\ref{alg:mst}. We give a deterministic algorithm to compute this integral in Section~\ref{sec:mst_combine}.

To derive our allocation, we use an \emph{auxiliary tree} that simulates Kruskal's algorithm, constructed as in Algorithm~\ref{alg:auxiliary}. The auxiliary tree is a rooted tree such that each leaf corresponds to a vertex in $V\cup \{r\}$.
We provide an overview of Algorithm~\ref{alg:auxiliary}. Initially, for each vertex $v \in V \cup \{r\}$, the algorithm prepares a vertex $u_{\{v\}}$ and sets its \emph{height} $h_{u_{\{v\}}}$ to $0$. The auxiliary tree is constructed by adding the edges of $E$ in ascending order of weight to graph $(V \cup \{r\},\emptyset)$. Edges of the same weight are added simultaneously. When adding the edges of a certain weight results in merging multiple connected components $C_1, \dots, C_k$ into a single connected component $C$, the algorithm creates a vertex $u_C$ corresponding to $C$ in the auxiliary tree. The height of $u_C$ is set as the weight of the edges at that time, and the edges are added to the auxiliary tree from $u_C$ to $u_{C_1}, \dots, u_{C_k}$.

\begin{algorithm}[t!]
\caption{Construction of the auxiliary tree}\label{alg:auxiliary}
\Procedure{\emph{\Call{AuxiliaryTree}{$G,\widehat{w}$}}}{
    \KwIn{A graph $G=(V\cup \{r\},E)$ and a weight vector $\widehat{w} \in \mathbb{R}_{\geq 0}^E$.}
    $U\leftarrow \{u_{\{v\}}\mid v\in V\cup \{r\}\}$, $F\leftarrow \emptyset$\;
    $h_u\leftarrow 0$ for each $u\in U$\;
    \For{$x\in \mathbb{R}_{\geq 0}$ such that $\mathcal{C}_{\widehat{w},<x}\neq \mathcal{C}_{\widehat{w},\leq x}$, in ascending order}{
        \For{$C\in \mathcal{C}_{\widehat{w},\leq x}\setminus \mathcal{C}_{\widehat{w},<x}$}{
            Add a new vertex $u_{C}$ to $U$\;
            $h_{u_{C}}\leftarrow x$\;
            \For{$C'\in \mathcal{C}_{\widehat{w},<x}$ such that $C'\subseteq C$}{
                Add a new edge $(u_{C},u_{C'})$ to $F$\;
            }
        }
    }
    \Return $(U,F,h)$\;
}
\end{algorithm}

For $x\in \mathbb{R}_{\geq 0}$ and a weight vector $\widehat{w}$, let $\mathcal{C}_{\widehat{w},<x}$ and $\mathcal{C}_{\widehat{w},\leq x}$ be the families of connected components of graphs whose vertex sets are $V\cup \{r\}$ and edge sets consist of edges $e\in E$ with $\widehat{w}_e<x$ and $\widehat{w}_e\leq x$, respectively. 
For the auxiliary tree $T$, we denote the subtree rooted at vertex $u$ by $T_u$. 
When the edges of an auxiliary tree are referred to as $(u, u')$, $u$ is the parent of $u'$.
For an edge $e=(u,u')$, we denote $T_e=T_{u'}$.
For simplicity, for $e=(u,u')\in E(T)$, we define $h_e := h_u$.

\Call{MSTGame}{$G,w,b$} first rounds each edge weight $w_e$ to a value $\widehat{w}_e$ that is proportional to a power of $2$, where the proportionality constant is determined by $b$. Let $T$ be the auxiliary tree derived from $(G,\widehat{w})$. Then, for each edge $e$ in $T$ such that $T_e$ does not have $r$ as a leaf, the value $h_e$ is evenly distributed among the agents corresponding to the leaves of $T_e$. 
At first glance, the total value distributed may seem unrelated to the value of the grand coalition.
However, it can be proved in Lemmas~\ref{lem:mst_core1} and~\ref{lem:mst_core2} that the total value distributed by this method is at least $\OPT(G,\widehat{w})$ and at most $2\OPT(G,\widehat{w})$.

\begin{algorithm}[t!]
\caption{Lipschitz continuous allocation for the minimum spanning tree game}\label{alg:mst}
\Procedure{\emph{\Call{MSTGame}{$G,w,b$}}}{
    \KwIn{A graph $G=(V\cup \{r\},E)$, a weight vector $w \in \mathbb{R}_{\geq 0}^E$, and $b\in [0,1]$.}
    For each $e\in E$ with $w_e>0$, let $\widehat{w}_e\leftarrow 2^{i_e+1+b}$, where $i_e$ be the unique integer such that $2^{i_e+b}\leq w_e < 2^{i_e+1+b}$\;
    $T\leftarrow \Call{AuxiliaryTree}{G,\widehat{w}}$ and identify the vertices of $G$ with the corresponding leaves of $T$\;
    \For{$e\in E(T)$ such that $r\not \in T_e$}{
        Let $X_{e}$ be the set of leaves in $T_{e}$\;
        Let $z_{e}\leftarrow \frac{h_e}{|X_{e}|}\bm{1}_{X_{e}}$\;
    }
    \Return $\sum_{e\in E}z_{e}$\;
}
\end{algorithm}

\subsection{Core Approximability}

We begin by analyzing the core approximability.
Let $T$ be the auxiliary tree for $(G,\widehat{w})$. 
For $X\subseteq V\cup \{r\}$, the \emph{connector} $\CONN(X)$ of $X$ is the minimal connected subgraph of $T$ that contains all leaves of $T$ corresponding to $X$.
The following lemma bounds the value of the characteristic function for a subset $S$ of $V$ using values that can be computed from connector $\CONN(S\cup \{r\})$. 
\begin{lemma}\label{lem:mst_aux}
Let $S\subseteq V$. Then, we have
\begin{align*}
    \sum_{\substack{(u,u')\in E(\CONN(S\cup \{r\}))\\ r\not \in T_{u'}}}(h_u-h_{u'})\leq \OPT(G[S\cup \{r\}], \widehat{w}).
\end{align*}
When $V=S$, the equality holds.
\end{lemma}
\begin{proof}
Let $R$ be the minimum spanning tree of $G[S\cup \{r\}]$ with respect to weight $\widehat{w}$. 
For $x\in \mathbb{R}_{\geq 0}$, let $N_{<x}=\left|\{C\in \mathcal{C}_{\widehat{w},<x}\colon C\cap (S\cup \{r\})\neq \emptyset\}\right|$ and $N_{\leq x}=\left|\{C\in \mathcal{C}_{\widehat{w},\leq x}\colon C\cap (S\cup \{r\})\neq \emptyset\}\right|$. Then, $T$ contains at least $N_{<x}-1$ edges $e$ with $\widehat{w}_e\geq x$. 
Specifically, we have
\begin{align*}
    \sum_{e\in E(R)}\widehat{w}_e\geq \sum_{\substack{x\in \mathbb{R}_{\geq 0}\\ \mathcal{C}_{\widehat{w},\leq x}\neq \mathcal{C}_{\widehat{w},<x}}}(N_{\leq x}-N_{<x})x
    = \sum_{u\in V(\CONN(S\cup \{r\}))}(d_{\CONN(S\cup \{r\})}(u)-1)h_u,
\end{align*}
where $d_{\CONN(S\cup \{r\})}(u)$ denotes the number of children of $u$ in $\CONN(S\cup \{r\})$. We observe that the equality holds for the first inequality for $V=S$ by recalling Kruskal's algorithm.
Let us prove
\begin{align*}
    \sum_{u\in V(\CONN(S\cup \{r\}))}(d_{\CONN(S\cup \{r\})}(u)-1)h_u = \sum_{\substack{(u,u')\in E(\CONN(S\cup \{r\}))\\ r\not \in T_{u'}}}(h_u-h_{u'}).
\end{align*}
to complete the proof. Because we have
\begin{align*}
       \sum_{\substack{(u,u')\in E(\CONN(S\cup \{r\}))\\ r\in T_{u'}}}(h_u-h_{u'}) = h_{\LCA(S\cup \{r\})},
\end{align*}
where $\LCA(C)$ denotes the lowest common ancestor of $C$ in the auxiliary tree, it is suffice to show
\begin{align*}
    \sum_{u\in V(\CONN(C))}(d_{\CONN(C)}(u)-1)h_u = \sum_{(u,u')\in E(\CONN(C))}(h_u-h_{u'}) - h_{\LCA(C)}
\end{align*}
holds for all $C\subseteq V\cup \{r\}$.
We prove this by induction.
For $|C|=1$, both the left- and the right-hand sides are equal to zero. For $|C|>1$, let $u_{C_1},\dots, u_{C_{d_{\CONN(C)}}}$ be children of $\LCA(C)$ in $\CONN(C)$. Then, we have
\begin{align*}
    &\sum_{u\in V(\CONN(C))}(d_{\CONN(C)}(u)-1)h_u\\
    &= \sum_{i=1}^{d_{\CONN(C)}}\sum_{u\in C_i}(d_{\CONN(C_i)}(u)-1)h_u + (d_{\CONN(C)}-1)h_{\LCA(C)}\\
    &= \sum_{i=1}^{d_{\CONN(C)}}\left(\sum_{(u,u')\in E(\CONN(C_i))}(h_u-h_{u'}) - h_{u_{C_i}}\right) + (d_{\CONN(C)}-1)h_{\LCA(C)}\\
    &=\sum_{i=1}^{d_{\CONN(C)}}\left(\sum_{(u,u')\in E(\CONN(C_i))}(h_u-h_{u'}) + (h_{\LCA(C)}-h_{u_{C_i}})\right) - h_{\LCA(C)}\\
    &= \sum_{(u,u')\in E(\CONN(C))}(h_u-h_{u'}) - h_{\LCA(C)},
\end{align*}
and the induction hypothesis is proved. Therefore, the lemma holds.
\end{proof}

Now, we have the following.
\begin{lemma}\label{lem:mst_core1}
Let $S\subseteq V$. Then, we have 
\begin{align*}
    \sum_{v\in S}\sum_{e\in E(T)}z_{e,v}\leq 4\OPT(G[S\cup \{r\}],w).
\end{align*}
\end{lemma}
\begin{proof}
For $v\in S$ and $e\in E(T)$ such that $r\not \in T_e$, $z_{e,v}>0$ happens only when $e\subseteq \CONN(S\cup \{r\})$. Therefore, we have
\begin{align*}
    \sum_{v\in S}\sum_{e\in E(T)}z_{e,v}
    &=\sum_{\substack{e\in E(\CONN(S\cup \{r\}))\\ r\not \in T_e}}\sum_{v\in S}z_{e,v}
    \leq \sum_{\substack{e\in E(\CONN(S\cup \{r\}))\\ r\not \in T_e}}\sum_{v\in V}z_{e,v}
    = \sum_{\substack{e\in E(\CONN(S\cup \{r\}))\\ r\not \in T_e}}h_e\\
    &\leq \sum_{\substack{(u,u')\in E(\CONN(S\cup \{r\}))\\ r\not \in T_e}}2(h_u-h_{u'})
    \leq 2\OPT(G[S\cup \{r\}],\widehat{w})\\
    &\leq 4\OPT(G[S\cup \{r\}],w),
\end{align*}
where the second inequality is from $h_{u}\geq 2h_{u'}$ that follows from the fact that each $\widehat{w}_e$ is rounded to a value proportional to the power of $2$ and the third inequality is from Lemma~\ref{lem:mst_aux}.
\end{proof}

We have the following bound for the grand coalition.
\begin{lemma}\label{lem:mst_core2}
We have 
\begin{align*}
\sum_{v\in V}\sum_{e\in E(T)}z_{e,v}\geq \OPT(G,w).
\end{align*}
\end{lemma}
\begin{proof}
We have
\begin{align*}
    \sum_{v\in V}\sum_{e\in E(T)}z_{e,v} = \sum_{e\in E(T)}h_e \geq \sum_{(u,u')\in E(T)}(h_u-h_{u'}) = \OPT(G,\widehat{w})\geq \OPT(G,w),
\end{align*}
where the second equality is obtained from Lemma~\ref{lem:mst_aux}.
\end{proof}

\subsection{Lipschitz Continuity}

Let $f\in E(G)$. We bound
\begin{align*}
    \frac{1}{\delta}\int_{0}^{1}\left\|\Call{MSTGame}{G,w,b} - \Call{MSTGame}{G,w+\delta \bm{1}_f,b}\right\|_1 \mathrm{d}b,
\end{align*}
which is an upper bound on
\begin{align*}
    \frac{1}{\delta}\left\|\int_{0}^{1}\Call{MSTGame}{G,w,b}\mathrm{d}b - \int_{0}^{1}\Call{MSTGame}{G,w+\delta \bm{1}_f,b}\mathrm{d}b\right\|_1.
\end{align*}
Without loss of generality, we can assume $\delta \leq w_f$ or $w_f=0$.
Now we fix $b\in [0,1]$. We denote the value of $T$, $\widehat{w}$, $X$, and $z$ in \Call{MSTGame}{$G,w,b$} (resp., \Call{MSTGame}{$G,w+\delta \bm{1}_f, b$}) as $T$, $\widehat{w}$, $X$, and $z$ (resp. $T'$, $\widehat{w}'$, $X'$, and $z'$).

When $\widehat{w}_f=\widehat{w}'_f$, $T$ and $T'$ are the same and \Call{MSTGame}{$G,w,b$} and \Call{MSTGame}{$G,w+\delta \bm{1}_f, b$} output the same vector. Otherwise, let $C$ be the connected component of $\mathcal{C}_{\widehat{w},\leq \widehat{w}_f}$ that contains both the endpoints of $f$. If $C$ is still connected even after the removal of $f$, then we have $\mathcal{C}_{\widehat{w},\leq x}=\mathcal{C}_{\widehat{w}',\leq x}$ for all $x\geq 0$ and thus $T=T'$.
Otherwise, let $C_1$ and $C_2$ be two connected components of $C$ after the removal of $f$. Subsequently, $T'$ is obtained from $T$ by the following operation:
\begin{itemize}
    \item[(1.1)] If no vertex $u_{C_1}$ exists in $T$, create a vertex $u_{C_1}$, set $h_{u_{C_1}}=\widehat{w}_f$, and for each child $u_X$ of $u_C$ with $X\subseteq C_1$, replace the edge $(u_C,u_X)$ with $(u_{C_1},u_X)$. Otherwise, delete the edge $(u_C,u_{C_1})$.
    \item[(1.2)] Do exactly the same for $C_2$. 
    \item[(2)] If $u_C$ has a parent $u_Y$ with $h_{u_Y}=\widehat{w}'_f$, delete the vertex $u_{C}$ and the edge $(u_Y,u_C)$, and add two new edges $(u_Y,u_{C_1})$ and $(u_Y,u_{C_2})$. Otherwise, add two new edges $(u_{C},u_{C_1})$ and $(u_{C},u_{C_2})$ and then change the value of $h_{u_C}$ from $\widehat{w}_f$ to $\widehat{w}'_f$.
\end{itemize}

We can observe that all edges $e$ of $T$ except for the edges $(u_C,u_{C_1})$ deleted in (1.1), $(u_C,u_{C_2})$ deleted in (1.2), and $(u_Y,u_C)$ deleted in (2) naturally correspond to edges in $T'$ that are not added in (2), and if we identify the edges of $T$ with those of $T'$ using that correspondence, it holds that $h_e = h'_e$ and $X_e = X'_e$, which implies $z_e=z'_e$. Therefore, we have the following.

\begin{lemma}
Assume $\widehat{w}_f\neq \widehat{w}'_f$. Then, we have
\begin{align*}
    \left\|\Call{MSTGame}{G,w,b} - \Call{MSTGame}{G,w+\delta \bm{1}_f,b}\right\|_1 \leq \widehat{w}_f + 2\widehat{w}'_f.
\end{align*}
\end{lemma}
\begin{proof}
Let
\begin{align*}
    z_1&=\begin{cases}
        \frac{1}{|C_1|}\bm{1}_{C_1} & \text{if $(u_C,u_{C_1})$ is deleted in (1.1) and $r\not \in C_1$}\\
        \bm{0} &\text{otherwise}        
    \end{cases},\\
    z_2&=\begin{cases}
        \frac{1}{|C_2|}\bm{1}_{C_2} & \text{if $(u_C,u_{C_2})$ is deleted in (1.2) and $r\not \in C_2$}\\
        \bm{0} &\text{otherwise}        
    \end{cases},\\
    z_3&=\begin{cases}
        \frac{1}{|C|}\bm{1}_{C} & \text{if $(u_Y,u_C)$ is deleted in (1.2) and $r\not \in C$}\\
        \bm{0} &\text{otherwise}        
    \end{cases},\\
    z_4&=\begin{cases}
        \frac{1}{|C_1|}\bm{1}_{C_1} & \text{if $r\not \in C_1$}\\
        \bm{0} &\text{otherwise}        
    \end{cases},\\
    z_5&=\begin{cases}
        \frac{1}{|C_2|}\bm{1}_{C_2} & \text{if $r\not \in C_2$}\\
        \bm{0} &\text{otherwise}        
    \end{cases}.
\end{align*}
Then, we have
\begin{align*}
    &\left\|\Call{MSTGame}{G,w,b} - \Call{MSTGame}{G,w+\delta \bm{1}_f,b}\right\|_1\\
    &= \left\|\left(z_1+z_2\right)\widehat{w}_f+\left(z_3-z_4-z_5\right)\widehat{w}'_f\right\|_1\\
    &= \left\|\left(z_1 \widehat{w}_f+\left(z_3\circ \bm{1}_{C_1}-z_4\right) \widehat{w}'_f\right)\right\|_1+\left\|\left(z_2 \widehat{w}_f+\left(z_3\circ \bm{1}_{C_2}-z_5\right) \widehat{w}'_f\right)\right\|_1\\
    &\leq \max\left(\widehat{w}_f+\frac{|C_1|}{|C|}\widehat{w}'_f, \widehat{w}'_f\right)+\max\left(\widehat{w}_f+\frac{|C_2|}{|C|}\widehat{w}'_f, \widehat{w}'_f\right)
    \leq \widehat{w}_f+2\widehat{w}'_f.\qedhere
\end{align*}
\end{proof}

The following lemma analyzes the probability that $\widehat{w}_f\neq \widehat{w}'_f$ happens.
\begin{lemma}
Assume $w'_f\leq 2w_f$. If $b$ is sampled uniformly from $[0,1]$, $\widehat{w}_f\neq \widehat{w}'_f$ happens with a probability of at most $\frac{\delta}{w_f\log 2}$.
\end{lemma}
\begin{proof}
$\widehat{w}_f\neq \widehat{w}'_f$ happens when there is an integer $i$ with $w_f<2^{i+b}\leq w_f+\delta$, implying that $\floor{\log_2 w_f-b}\neq \floor{\log_2 w_f+\delta-b}$. This happens with probability 
\[
    \log_2 (w_f+\delta) - \log_2 w_f = \log_2 \left(1+\frac{\delta}{w_f}\right)\leq \frac{\delta}{w_f \log 2}.\qedhere
\]
\end{proof}

Now, we have the following:
\begin{lemma}\label{lem:mst_lip}
We have
\begin{align*}
    \frac{1}{\delta}\int_{0}^{1}\left\|\Call{MSTGame}{G,w,b} - \Call{MSTGame}{G,w+\delta \bm{1}_f,b}\right\|_1 \mathrm{d}b\leq \frac{10\delta}{\log 2}.
\end{align*}
\end{lemma}
\begin{proof}
If $w_f=0$, we have
\begin{align*}
    \frac{1}{\delta}\int_{0}^{1}\left\|\Call{MSTGame}{G,w,b} - \Call{MSTGame}{G,w+\delta \bm{1}_f,b}\right\|_1 \mathrm{d}b
    \leq 2\widehat{w}'_f \leq 4w'_f = 4\delta.
\end{align*}
Otherwise, we have
\begin{align*}
    &\frac{1}{\delta}\int_{0}^{1}\left\|\Call{MSTGame}{G,w,b} - \Call{MSTGame}{G,w+\delta \bm{1}_f,b}\right\|_1 \mathrm{d}b\\
    &\leq \frac{\delta}{w_f\log 2}\left(\widehat{w}_f+2\widehat{w}'_f\right) 
    = \frac{\delta}{w_f\log 2}\cdot 5\widehat{w}_f 
    \leq \frac{\delta}{w_f\log 2}\cdot 10w_f = \frac{10\delta}{\log 2}.
 \end{align*}
\end{proof}

\subsection{Proof of Theorem~\ref{thm:mst}}\label{sec:mst_combine}
Combining Lemmas~\ref{lem:mst_core1},~\ref{lem:mst_core2}, and~\ref{lem:mst_lip} and applying Lemma~\ref{lem:benri_cost} yields Theorem~\ref{thm:mst}.
\begin{proof}[Proof of Theorem~\ref{thm:mst}]
Combining Lemmas~\ref{lem:mst_core1},~\ref{lem:mst_core2},~\ref{lem:mst_lip}, and~\ref{lem:benri_cost}, we obtain an algorithm that returns $4$-approximate core allocation for the matching game with Lipschitz constant $\frac{20}{\log 2}+1=O(1)$.
The fact that the allocation defined by Lemma~\ref{lem:match_final} can be computed in polynomial time is obtained using the same argument as in the proof of Theorem~\ref{thm:matching}.
\end{proof}

%% file: shapley.tex
\section{Shapley Value}\label{sec:shapley}
In this section, we discuss the Lipschitz continuity of the Shapley values of the optimization games and prove Theorems~\ref{thm:shap_matching} and~\ref{thm:shap_mst}. 

\subsection{Matching Game}

Let $n$ be an odd integer greater than or equal to $5$. Consider a path $G=(V,E)$ consisting of $n$ vertices labeled sequentially as $v_1, \dots, v_{n}$. 
Let $\SHAPMATCH(w)$ be the Shapley value for the matching game with respect to weight $w$.
The weight vectors $w, w' \in \mathbb{R}_{\geq 0}^{E}$ are defined as
\begin{align*}
    w_{(v_i,v_{i+1})}&=1 \quad (i=1,\dots, n-1),\\
    w'_{(v_i,v_{i+1})}&=\begin{cases}
        1+\delta & \text{if }i = 2\\
        1 & \text{otherwise}.
    \end{cases}
\end{align*}
For this instance, we prove
\begin{align*}
    \|\SHAPMATCH(w)-\SHAPMATCH(w')\|\geq \Omega(\delta \log n),
\end{align*}
which directly implies Theorem~\ref{thm:shap_matching}.

For $\sigma\in \mathfrak{S}_V$ and $v\in V$ such that $\sigma(k)=v$, let 
\begin{align*}
    x_{\sigma,v}&=\OPT(\{\sigma(1),\dots,\sigma(k)\},w)-\OPT(\{\sigma(1),\dots, \sigma(k-1)\},w),\\
    x'_{\sigma,v}&=\OPT(\{\sigma(1),\dots,\sigma(k)\},w')-\OPT(\{\sigma(1),\dots, \sigma(k-1)\},w').
\end{align*}

\begin{lemma}\label{lem:shap_match_frac}
Let $i\geq 4$ be an even number. If $\sigma$ is uniformly sampled from $\mathfrak{S}_V$, then we have
\begin{align*}
    \left|\E\left[x_{\sigma,v_i}\right]-\E\left[x'_{\sigma,v_i}\right]\right|\geq \frac{\delta}{i+1}.
\end{align*}
\end{lemma}
\begin{proof}
Let $k=\sigma^{-1}(v_i)$ and $S=\{v\in V\colon \sigma^{-1}(v)<k\}$.
Let $c_1$ (resp., $c_2$) be the size of the connected component of graph $G[S]$ containing $v_{i-1}$ (resp., $v_{i+1}$). If $v_{i-1}\not \in S$ (resp., $v_{i+1}\not \in S$), we set $c_1=0$ (resp., $c_2=0$).

Clearly, when the connected component of $G[S]$ containing $v_{i-1}$ does not include edge $(v_2,v_3)$, that is, when $c_1 \leq i-3$, it holds that $x_{\sigma,v_i} = x'_{\sigma,v_i}$. When $c_1 = i-2$, we have
\begin{align*}
    x_{\sigma,v_i}&=\begin{cases}
        1 & \text{if $c_2$ is odd}\\
        0 & \text{otherwise}             
    \end{cases}\\
    x'_{\sigma,v_i}&=\begin{cases}
        1 & \text{if $c_2$ is odd}\\
        0 & \text{otherwise}\\
    \end{cases}
\end{align*}
and thus $x_{\sigma,v_i}=x'_{\sigma,v_i}$. When $c_1=i-1$, we have
\begin{align*}
    x_{\sigma,v_i}&=\begin{cases}
        1 & \text{if $c_2$ is even}\\
        0 & \text{otherwise}             
    \end{cases}\\
    x'_{\sigma,v_i}&=\begin{cases}
        1-\delta & \text{if $c_2$ is even}\\
        0 & \text{otherwise}\\
    \end{cases}
\end{align*}
Therefore, $x_{\sigma,v_i} \neq x'_{\sigma,v_i}$ happens only when $c_1 = i-1$ and $c_2$ is even. In such cases, we have $x_{\sigma,v_i} - x'_{\sigma,v_i} = \delta$.
Now, we have
\begin{align*}
    \left|\E\left[x_{\sigma,v_i}\right]-\E\left[x'_{\sigma,v_i}\right]\right|
    &= \Pr\left[x_{\sigma,v_i}\neq x'_{\sigma,v_i}\right]\cdot \delta\\
    &= \Pr\left[c_1=i-1 \land c_2 \equiv 0 \pmod 2\right]\cdot \delta\\
    &\geq \Pr\left[c_1=i-1 \land c_2 = 0\right]\cdot \delta\\
    &= \Pr\left[\max(\sigma^{-1}(v_1),\dots, \sigma^{-1}(v_{i+1}))=\sigma^{-1}(v_i)\right]\cdot \delta\\
    &= \frac{\delta}{i+1}.\qedhere
\end{align*}
\end{proof}
Theorem~\ref{thm:shap_matching} is derived directly from the following lemma:
\begin{lemma}
We have
\begin{align*}
    \|\SHAPMATCH(w)-\SHAPMATCH(w')\|\geq \Omega(\delta \log n),
\end{align*}
\end{lemma}
\begin{proof}
From Lemma~\ref{lem:shap_match_frac}, we have
\begin{align*}
    \|\SHAPMATCH(w)-\SHAPMATCH(w')\|
    &= \sum_{i=1}^{n}\left|\E\left[x_{\sigma,v_i}\right]-\E\left[x'_{\sigma,v_i}\right]\right|\\
    &\geq \sum_{i\in \{4,6,8,\dots, n-1\}}\left|\E\left[x_{\sigma,v_i}\right]-\E\left[x'_{\sigma,v_i}\right]\right|\\
    &\geq \sum_{i\in \{4,6,8,\dots, n-1\}}\frac{\delta}{i+1}
    \geq \Omega(\delta \log n).\qedhere
\end{align*}
\end{proof}

\subsection{Minimum Spanning Tree Game}

We begin with the following:
\begin{lemma}\label{lem:mstshapdiff}
Let $f\in E$, $S\subseteq V$, and $v\in V$ such that $f\subseteq S\subseteq V\setminus \{v\}$.
Let $\delta>0$ and $w'=w+\delta \bm{1}_f$.
Then, we have
\begin{align*}
    \OPT(S\cup \{v,r\},w')-\OPT(S\cup \{v,r\},w)\leq \OPT(S\cup \{r\},w')-\OPT(S\cup \{r\},w).
\end{align*}
\end{lemma}
\begin{proof}
Let $T$ and $R$ be the minimum spanning trees with respect to weight $w$ for $G[S\cup \{r\}]$ and $G[S\cup \{v,r\}]$, respectively. 
Similarly, let $T'$ and $R'$ be the minimum spanning trees with respect to weight $w'$ for $G[S\cup \{r\}]$ and $G[S\cup \{v,r\}]$, respectively. 
If $f$ is not an edge of $R$, the left-hand side is zero; thus, the lemma holds. 
Assume otherwise.
Let the connected components of the graph obtained by removing $v$ and its incident edges from $R$ be $C_1, \dots, C_k$. Because $v$ is not an endpoint of $f$ by the assumption of the lemma, $f$ is the edge of one of the $C_i$'s. Without loss of generality, we assume that this is an edge of $C_1$. Let $D_1$ and $D_2$ be the two connected components formed by removing $f$ from $C_1$, where $D_1$ contains the neighbor of $v$ on $R$, and $D_2$ does not.

$T$ is formed by connecting trees $C_1, \dots, C_k$ with $k-1$ edges. In particular, $f$ is an edge of $T$. $T'$ is obtained by removing $f$ from $T$ and adding an edge $f'$ (which may be the same as $f$).
We show that an edge $f''$ with $w'_{f''} \leq w'_{f'}$ exists such that $R\setminus \{f\}\cup \{f''\}$ is a tree. If this is shown, then we have
\begin{align*}
&\OPT(S\cup \{v,r\},w')-\OPT(S\cup \{v,r\},w)\\
&\leq w'_{f''} - w_f
\leq w'_{f'}-w_f
= \OPT(S\cup \{r\},w')-\OPT(S\cup \{r\},w)
\end{align*}
and the lemma is proved.

If one endpoint of $f'$ belongs to $D_2$, then, by the connectivity of $T'$, the other endpoint of $f'$ does not belong to $D_2$. Therefore, by setting $f''=f'$, $v$ and $D_2$ are connected on $R\setminus \{f\}\cup \{f''\}$, which implies it forms a tree, and $w'_{f''}=w'_{f'}$ holds. 
Otherwise, let $P$ be the unique path on $T'$ such that it has two endpoints in $D_1$ and $D_2$, and contains no other vertices in $D_1 \cup D_2 = C_1$. Let $f''$ be the unique edge on $P$ with an endpoint in $D_2$. Because $f''$ is not contained in $C_1$, we have $f \neq f''$. As $v$ and $D_2$ are connected on $R\setminus \{f\}\cup \{f''\}$, it forms a tree. 
Furthermore, $T\setminus \{f''\}\cup \{f'\}$ is a tree. Because $T$ is the minimum spanning tree with respect to weight $w$, it follows that $w'_{f''}=w_{f''}\leq w_{f'}$. Hence, the lemma is proved.
\end{proof}

Let $\SHAPMST(w)$ be the Shapley value for the minimum spanning tree game with respect to weight $w$. 
Theorem~\ref{thm:shap_mst} is proved by combining the next lemma and Lemma~\ref{lem:seeoneelement}.
\begin{lemma}
Let $\delta > 0$ and $f\in E$. Then, we have
\begin{align*}
    \left\|\SHAPMST(w)-\SHAPMST(w+\delta \bm{1}_f)\right\|_1\leq 2\delta.    
\end{align*}
\end{lemma}
\begin{proof}
For $\sigma \in \mathfrak{S}_{V}$ and $k\in \{1,\dots, |V|\}$, let
\begin{align*}
    x_{\sigma,k}&=\OPT(\{r,\sigma(1),\dots, \sigma(k)\},w)-\OPT(\{r,\sigma(1),\dots, \sigma(k-1)\},w),\\
    x'_{\sigma,k}&=\OPT(\{r,\sigma(1),\dots, \sigma(k)\},w+\delta \bm{1}_f)-\OPT(\{r,\sigma(1),\dots, \sigma(k-1)\},w+\delta \bm{1}_f).
\end{align*}
Let $i$ be the first index, such that $f\subseteq \{r,\sigma(1),\dots, \sigma(i)\}$. Then, for $j<i$, we have $x_{\sigma, j}=x'_{\sigma, j}$. We also have
\begin{align*}
    x'_{\sigma,i}-x_{\sigma,i} = \OPT(\{r,\sigma(1),\dots, \sigma(i)\},w+\delta \bm{1}_f)-\OPT(\{r,\sigma(1),\dots, \sigma(i)\},w)\leq \delta.
\end{align*}
Furthermore, Lemma~\ref{lem:mstshapdiff} implies that for $j>i$, we have $x'_{\sigma,j}-x_{\sigma,j}\leq 0$.
Finally, we have
\begin{align*}
    \sum_{j=1}^{|V|}\left(x'_{\sigma,j}-x_{\sigma,j}\right)
    = \OPT(V\cup \{r\},w+\delta \bm{1}_f)-\OPT(V\cup \{r\},w)\geq 0.
\end{align*}
Thus, we have
\begin{align*}
    \sum_{j=1}^{|V|}\left|x'_{\sigma,j}-x_{\sigma,j}\right|
    \leq \delta + \left(\delta - \sum_{j=1}^{|V|}\left(x'_{\sigma,j}-x_{\sigma,j}\right)\right)
    \leq 2\delta.
\end{align*}
Therefore, we have
\[
    \left\|\SHAPMST(w)-\SHAPMST(w')\right\|_1
    \leq \frac{1}{|V|!}\sum_{\sigma \in \mathfrak{S}_{V}}\sum_{j=1}^{|V|}\left|x'_{\sigma,j}-x_{\sigma,j}\right|
    \leq 2\delta. \qedhere
\]
\end{proof}

%% file: omitted_proofs.tex
\section{Omitted Proofs}\label{sec:omitted}

\begin{proof}[Proof of Lemma~\ref{lem:benri_welfare}]
Let $x(w)=\frac{\nu(V,w)}{\|\mathcal{A}(w)\|_1}\mathcal{A}(w)$. Then, $x(w)$ is in the $\frac{1}{D}$-approximate core because $\|x(w)\|_1=\nu(V,w)$ and
\begin{align*}
    \sum_{v\in S}x(w)_v=\frac{\nu(V,w)}{\|\mathcal{A}(w)\|_1}\sum_{v\in S}\mathcal{A}(w)_v\geq \frac{1}{D}\sum_{v\in S}\mathcal{A}(w)_v\geq \frac{1}{D}\nu(S,w).
\end{align*}
Furthermore, we have
\begin{align*}
    &\|x(w)-x(w')\|_1 \\
    &= \left\|\frac{\nu(V,w)}{\|\mathcal{A}(w)\|_1}\mathcal{A}(w)-\frac{\nu(V,w')}{\|\mathcal{A}(w')\|_1}\mathcal{A}(w')\right\|_1\\
    &\leq \left\|\frac{\nu(V,w)}{\|\mathcal{A}(w)\|_1}\mathcal{A}(w)-\frac{\nu(V,w)}{\|\mathcal{A}(w)\|_1}\mathcal{A}(w')\right\|_1
    +\left\|\frac{\nu(V,w)}{\|\mathcal{A}(w)\|_1}\mathcal{A}(w')-\frac{\nu(V,w)}{\|\mathcal{A}(w')\|_1}\mathcal{A}(w')\right\|_1\\
    &\quad + \left\|\frac{\nu(V,w)}{\|\mathcal{A}(w')\|_1}\mathcal{A}(w')-\frac{\nu(V,w')}{\|\mathcal{A}(w')\|_1}\mathcal{A}(w')\right\|_1\\
    &= \frac{\nu(V,w)}{\|\mathcal{A}(w)\|_1}\cdot \|\mathcal{A}(w)-\mathcal{A}(w')\|_1 
    + \left|\frac{\|\mathcal{A}(w')\|_1-\|\mathcal{A}(w)\|_1}{\|\mathcal{A}(w)\|_1}\right|\nu(V,w)
    + |\nu(V,w)-\nu(V,w')|\\
    &\leq 1\cdot \|\mathcal{A}(w)-\mathcal{A}(w')\|_1 + 1\cdot \|\mathcal{A}(w)-\mathcal{A}(w')\|_1 + \|w-w'\|_1\\
    &\leq 2L\|w-w'\|_1+\|w-w'\|_1=(2L+1)\|w-w'\|_1.
    \qedhere
\end{align*}
\end{proof}

\begin{proof}[Proof of Lemma~\ref{lem:benri_cost}]
Let $x(w)=\frac{\nu(V,w)}{\|\mathcal{A}(w)\|_1}\mathcal{A}(w)$. Then, $x(w)$ is in the $D$-approximation core because $\|x(w)\|_1=\nu(V,w)$ and
\begin{align*}
    \sum_{v\in S}x(w)_v=\frac{\nu(V,w)}{\|\mathcal{A}(w)\|_1}\sum_{v\in S}\mathcal{A}(w)_v\leq \sum_{v\in S}\mathcal{A}(w)_v\leq D\nu(S,w).
\end{align*}
Furthermore, we have
\begin{align*}
    &\|x(w)-x(w')\|_1 \\
    &= \left\|\frac{\nu(V,w)}{\|\mathcal{A}(w)\|_1}\mathcal{A}(w)-\frac{\nu(V,w')}{\|\mathcal{A}(w')\|_1}\mathcal{A}(w')\right\|_1\\
    &\leq \left\|\frac{\nu(V,w)}{\|\mathcal{A}(w)\|_1}\mathcal{A}(w)-\frac{\nu(V,w)}{\|\mathcal{A}(w)\|_1}\mathcal{A}(w')\right\|_1
    +\left\|\frac{\nu(V,w)}{\|\mathcal{A}(w)\|_1}\mathcal{A}(w')-\frac{\nu(V,w)}{\|\mathcal{A}(w')\|_1}\mathcal{A}(w')\right\|_1\\
    &\quad + \left\|\frac{\nu(V,w)}{\|\mathcal{A}(w')\|_1}\mathcal{A}(w')-\frac{\nu(V,w')}{\|\mathcal{A}(w')\|_1}\mathcal{A}(w')\right\|_1\\
    &= \frac{\nu(V,w)}{\|\mathcal{A}(w)\|_1}\cdot \|\mathcal{A}(w)-\mathcal{A}(w')\|_1 
    + \left|\frac{\|\mathcal{A}(w')\|_1-\|\mathcal{A}(w)\|_1}{\|\mathcal{A}(w)\|_1}\right|\nu(V,w) 
    + |\nu(V,w)-\nu(V,w')|\\
    &\leq 1\cdot \|\mathcal{A}(w)-\mathcal{A}(w')\|_1 + 1\cdot \|\mathcal{A}(w)-\mathcal{A}(w')\|_1 + \|w-w'\|_1\\
    &\leq 2L\|w-w'\|_1+\|w-w'\|_1=(2L+1)\|w-w'\|_1.
    \qedhere
\end{align*}
\end{proof}